 \documentclass[preprint,review,12pt]{elsarticle}

\usepackage{amssymb}
\usepackage{amsthm}

\newcommand{\scriptf}{\mathcal{F}}
\newcommand{\scripte}{\mathcal{E}}
\newcommand{\scriptv}{\mathcal{V}}
\newcommand{\sv}{{\mathcal V}}
\newcommand{\se}{{\mathcal E}}

\usepackage{graphicx}

\usepackage{subfigure}
\usepackage{amsmath}
\usepackage{amssymb}
\usepackage{mdwlist}
 
\usepackage{algorithm}
\usepackage{hyperref}
\usepackage{xspace}
\usepackage{setspace}

\newtheorem{theorem}{Theorem}

\newtheorem{claim}{Claim}

\newtheorem{definition}{Definition}

\journal{Computer Science}

\begin{document}

\begin{frontmatter}


\title{Broadcast Using Certified Propagation Algorithm in Presence of Byzantine Faults\tnoteref{label1}}
\tnotetext[label1]{\normalsize This research is supported in part by Army Research Office grant W-911-NF-0710287. Any opinions, findings, and conclusions or recommendations expressed here are those of the authors and do not necessarily reflect the views of the funding agencies or the U.S. government.}

\author{Lewis Tseng\corref{cor1}\fnref{label2}}
\ead{ltseng3@illinois.edu}
\author[label3]{Nitin H. Vaidya}
\ead{nhv@illinois.edu}
\author[label4]{Vartika Bhandari}
\ead{vartika@google.com}

\fntext[label2]{Department of Computer Science and Coordinated Science Laboratory, University of Illinois at Urbana-Champaign}
\fntext[label3]{Department of Electrical and Computer Engineering and Coordinated Science Laboratory, University of Illinois at Urbana-Champaign}
\fntext[label4]{Google, Inc.}
\cortext[cor1]{Corresponding author.}





\begin{abstract}
We explore the correctness of the Certified Propagation Algorithm (CPA) \cite{Koo_broadcast,Vartika_broadcast, Pelc_broadcast, Ichimura_broadcast} in solving broadcast with locally bounded Byzantine faults. CPA allows the nodes to use only local information regarding the network topology. We provide a {\em tight} necessary and sufficient condition on the network topology for the correctness of CPA. To the best of our knowledge, this work is the first to solve the open problem in \cite{Pelc_broadcast}. We also present some simple extensions of this result.

\end{abstract}

\begin{keyword}

Algorithms \sep Distributed computing \sep Fault tolerance \sep Broadcast \sep Byzantine faults \sep Tight condition


\end{keyword}

\end{frontmatter}


\section{Introduction}

In this work, we explore fault-tolerant broadcast with locally bounded Byzantine faults in synchronous point-to-point networks. We assume a {\em $f$-locally bounded model}, in which at most $f$ Byzantine faults occur in the neighborhood of every {\em fault-free} node \cite{Koo_broadcast}. In particular, we are interested in the necessary and sufficient condition on the underlying communication network topology for the correctness of the Certified Propagation Algorithm (CPA) -- the CPA algorithm has been analyzed in prior work \cite{Koo_broadcast,Vartika_broadcast, Pelc_broadcast, Ichimura_broadcast}. 

\paragraph{Problem Formulation}
Consider an arbitrary directed network of $n$ nodes. One node in the network, called the {\em source} ($s$), is given an initial input, which the source node needs to transmit to all the other nodes. The source $s$ is assumed to be {\em fault-free}. We say that CPA is {\em correct}, if it satisfies the following properties, where $x_s$ denotes the input at source node $s$:

\begin{itemize}
\item {\bf Termination:} every fault-free node $i$ eventually decides on an output value $y_i$.


\item {\bf Validity:} for every fault-free node $i$, its output value $y_i$ equals the source's input, i.e., $y_i = x_s$. As stated above, the source node is assumed to be fault-free.

\end{itemize}

In this paper, we study the condition on the network topology for the correctness of CPA.

\paragraph{Related Work}
Several researchers have addressed CPA problem.
\cite{Koo_broadcast} studied the problem in an infinite grid. \cite{Vartika_broadcast} developed a sufficient condition in the context of arbitrary network topologies, but the sufficient condition proposed is not tight. \cite{Pelc_broadcast} provided necessary and sufficient conditions, but the two conditions are not identical (not tight). \cite{Ichimura_broadcast} provided another condition that can approximate (within a factor of 2) the largest $f$ for which CPA is correct in a given graph.

Independently, conditions similar to our condition are also discovered by other researchers \cite{Secret, network} under other contexts. \cite{Secret} proved a similar condition to be sufficient (but not tight) to solve Shamir's $(n, k)$ threshold secret sharing problem, where the source wants to transmit shares of secret to all the other nodes, and all nodes are assumed to be honest-but-curious. In the context of cascading behavior in the network, \cite{network} showed that a similar condition is necessary and sufficient to achieve a complete cascade, i.e., all nodes have learned the value transmitted by a cluster of sources with same input values using only local information. In their model, all nodes are assumed to be fault-free. Due to our assumption of existence of Byzantine failures, the proofs in this paper are different from the ones in \cite{Secret, network}.

\section{System Model}
\label{sec:models}

The system is assumed to be {\em synchronous}. The synchronous communication network consisting of $n$ nodes including source node $s$ is modeled as a simple {\em directed} graph $G(\scriptv,\scripte)$, where $\scriptv$ is the set of $n$ nodes, and $\scripte$ is the set of directed edges between the nodes in $\scriptv$.  We assume that $n\geq 2$, since the problem for $n=1$ is trivial. Node $i$ can transmit messages to another node $j$ if and only if the directed edge $(i,j)$ is in $\scripte$. Each node can transmit messages to itself as well; however, for convenience, we {exclude self-loops} from set $\scripte$. That is, $(i,i)\not\in\scripte$ for $i\in\scriptv$. All the links (i.e., communication channels) are assumed to be reliable, FIFO (first-in first-out) and deliver each transmitted message exactly once. With a slight abuse of terminology, we will use the terms {\em edge} and {\em link} interchangeably.

For each node $i$, let $N_i^-$ be the set of nodes from which $i$ has incoming
edges. That is, $N_i^- = \{\, j ~|~ (j,i)\in \scripte\, \}$. Similarly, define $N_i^+$ as the set of nodes to which node $i$ has outgoing edges. That is, $N_i^+ = \{\, j ~|~ (i,j)\in \scripte\, \}$. Nodes in $N_i^-$ and $N_i^+$ are, respectively, said to be incoming and outgoing neighbors of node $i$. Since we exclude self-loops from $\scripte$, $i\not\in N_i^-$ and $i\not\in N_i^+$.  However, we note again that each node can indeed transmit messages to itself.

We consider the $f$-local fault model, with at most $f$ incoming neighbors of any fault-free node becoming Byzantine faulty. \cite{Koo_broadcast, Vartika_broadcast, Pelc_broadcast, Ichimura_broadcast} also studied CPA problem under this fault model. Yet, to the best of our knowledge, the tight necessary and sufficient conditions for the correctness of CPA in synchronous arbitrary point-to-point networks under $f$-local fault model have not been developed previously.

\section{Feasibility of CPA under $f$-local fault model}

\subsection{Certified Propagation Algorithm (CPA)}
\label{sec:cpa}

In this subsection, we describe the Certified Propagation Algorithm (CPA) from \cite{Koo_broadcast} formally. Note that the faulty nodes may deviate from this specification arbitrarily. Possible misbehavior includes sending incorrect and mismatching (or inconsistent) messages to different outgoing neighbors.

Source node $s$ commits to its input $x_s$ at the start of the algorithm,
i.e., sets its output equal to $x_s$.
The source node is said to have committed to $x_s$ in round 0.
The algorithm for each round $r~(r>0)$, is as follows:

\begin{enumerate}
\item Each node that commits in round $r-1$ to some value $x$,
transmits message $x$ to all its outgoing neighbors, and then terminates.

\item If any node receives message $x$ directly from source $s$, it commits to output $x$.

\item Through round $r$, if a node has received messages containing value $x$ from at least $f+1$ distinct incoming neighbors, then it commits to output $x$.
\end{enumerate}

\subsection{The Necessary Condition}
\label{s:condition}

For CPA to be correct, the network graph $G(\sv, \se)$ must satisfy the necessary condition proved in this section. We borrow two relations $\Rightarrow$ and $\not\Rightarrow$ from our previous paper \cite{vaidya_PODC12}.

\begin{definition}
For non-empty disjoint sets of nodes $A$ and $B$,

\begin{itemize}
\item $A \Rightarrow B$ iff there exists a node $v \in B$ that has at least $f+1$ distinct incoming neighbors in $A$, i.e., $|N_v^- \cap A| > f$.
\item $A \not\Rightarrow B$ iff $A \Rightarrow B$ is not true.
\end{itemize}
\end{definition}

\begin{definition}
Set $F \subseteq \sv$ is said to be a \underline{feasible} $f$-local fault set, if for each node $v \not\in F$, $F$ contains at most $f$ incoming neighbors of node $v$. That is, for every $v \in \sv - F, |N_v^- \cap  F| \leq f$.
\end{definition}

We now derive the necessary condition on the network topology.

\begin{theorem}
\label{thm:nec}
Suppose that CPA is correct in graph $G(\sv, \se)$ under the $f$-local fault model. Let sets $F, L, R$ form a partition\footnote{Sets $X_1,X_2,X_3,...,X_p$ are said to form a partition of set $X$ provided that (i) $\cup_{1\leq i\leq p} X_i = X$, and (ii) $X_i\cap X_j=\Phi$ if $i\neq j$.} of $\sv$, such that (i) source $s\in L$, (ii) $R$ is non-empty, and (iii) $F$ is a feasible $f$-local fault set. Then

\begin{itemize}
\item $L \Rightarrow R$, or

\item $R$ contains an outgoing neighbor of $s$, i.e., $N_s^+ \cap R \neq \emptyset$.

\end{itemize}

\end{theorem}

\begin{proof}
Consider any partition $F, L, R$ such that $s\in L$, $R$ is non-empty, and $F$ is a feasible $f$-local fault set. Suppose that the input at $s$ is $x_s$. Consider any single execution of the CPA algorithm such that the nodes in $F$ behave as if they have crashed.

By assumption, CPA is correct in the given network under such a behavior by the faulty nodes. Thus, all the fault-free nodes eventually commit their output to $x_s$. Let round $r~(r>0)$, be the earliest round in which at least one of the nodes in $R$ commits to $x_s$. Let $v$ be one of the node in $R$ that commits in round $r$. Such a node $v$ must exist since $R$ is non-empty, and it does not contain source node $s$. For node $v$ to be able to commit, as per specification of the CPA algorithm, either node $v$ should receive the message $x_s$ directly from the source $s$, or node $v$ must have $f + 1$ distinct incoming neighbors that have already committed to $x_s$. By definition of node $v$, nodes that have committed to $x_s$ prior to v must be outside $R$; since nodes in $F$ behave as crashed, these $f+1$ nodes must be in $L$. Thus, either $(s,v) \in \se$, or node $v$ has at least $f + 1$ distinct incoming neighbors in set $L$.


\end{proof}

\subsection{Sufficiency}

We now show that the condition in Theorem \ref{thm:nec} is also sufficient.

\begin{theorem}
\label{thm:suf}
If $G(\sv, \se)$ satisfies the condition in Theorem \ref{thm:nec},
then CPA is correct in $G(\sv,\se)$ under the $f$-local fault model.
\end{theorem}

\begin{proof}

Suppose that $G(\sv, \se)$ satisfies the condition in Theorem \ref{thm:nec}.
Let $F'$ be the set of faulty nodes. By assumption, $F'$ is a feasible local fault set.
Let $x_s$ be the input at source node $s$.
We will show that, (i) fault-free nodes do not commit to any value other
than $x_s$ (Validity),
and, (ii) until all the fault-free nodes have committed, in each round of CPA,
at least one additional fault-free node commits to value $x_s$ (Termination).
The proof is by induction.

~

\noindent {\em Induction basis:} Source node $s$ commits in round 0
to output equal to its input $x_s$. No other fault-free nodes commit in
round 0.

\noindent{\em Induction:}
Suppose that $L$ is the set of fault-free nodes that have committed
to $x_s$ through round $r$, $r\geq 0$. Thus, $s\in L$.
Define $R=\sv-L-F'$.
If $R=\emptyset$, then the proof is complete. Let us now assume that
$R\neq\emptyset$.

Now consider round $r+1$.

\begin{itemize}
\item Validity:

Consider any fault-free node $u$ that has not committed prior to round $r+1$
(i.e., $u\in R$). All the nodes in $L$ have committed to $x_s$ by the
end of round $r$. Thus,
in round $r+1$ or earlier, node $u$ may receive messages containing values
different from $x_s$ only from nodes in $F'$. Since there are at most $f$
incoming neighbors of $u$ in $F'$, node $u$ cannot commit to any value different
from $x_s$ in round $r+1$. 

\item Termination:

By the condition in Theorem \ref{thm:nec}, there exists a node $w$
in $R$ such that (i) node $w$ has an incoming link from $s$, or (ii)
node $w$ has incoming links from $f+1$ nodes in $L$.
In case (i), node $w$ will commit to $x_s$ on receiving $x_s$ from node $s$
in round $r+1$ (in fact, $r+1$ in this case must be $1$).
In case (ii), since all the nodes in $L$ from whom node $w$ has
incoming links have committed to $x_s$ (by definition of $L$), node $w$
will be able to commit to $x_s$ after receiving messages from at least
$f+1$ incoming neighbors in $L$, since all nodes in $L$ have committed to $x_s$
by the end of round $r$ by the definition of $L$.\footnote{Since node $w$ did not commit prior to round $r+1$, it follows that at least one node in $L$ must have committed in round $r$.} Thus, node $w$ will commit to $x_s$ in round $r+1$.

\end{itemize}

This completes the proof.
\end{proof}

\section{CPA without prior knowledge of $f$}

In this section, we present a parameter-independent algorithm CPA-P that does not require prior knowledge of $f$, and each node only needs to know $n$, the number of nodes in the system. That is, given a graph $G$ that can tolerate $f$-local faults (where $f$ is unknown), the algorithm CPA-P presented below solves the broadcast problem in $G$ without usage of $f$. 


The core idea of CPA-P is for each node to exhaustively test all possible parameters by running $n+1$ instances of CPA algorithm in parallel. Each instance of CPA algorithm corresponds to a tested parameter ranging from $0$ to $n$. That is, each instance assumes that the tested parameter is the real bound ($f$) on the local faults at each node.\footnote{For simplicity of presentation, we assume that every node keeps track of $n+1$ instances (of the CPA algorithm) at the same time, even if the node already knows that some instances cannot terminate, since it may never receive enough identical messages if the tested parameter is too large. In a real implementation, each node $i$ only needs to keep track of $\lceil \frac{d_i}{2} \rceil - 1$ instances of CPA algorithm, where $d_i$ is the number of incoming neighbors at node $i$.} The correctness of CPA-P is based on the following observation: For each fault-free node, when the tested parameter is larger than or equal to the real parameter $f$, then there are only two outcomes: (i) it cannot commit, since it did not receive enough identical messages (violating Step 3 in CPA as specified in \ref{sec:cpa}), or (ii) it commits to a correct value, i.e., the input value of the source. Thus, in the end of the CPA-P,\footnote{Note that CPA is guaranteed to terminate in $n$ steps, and so is CPA-P.} each node can simply commit to the non-null value corresponding to the largest tested parameter. Now, we describe CPA-P formally.

Throughout the execution, each node $i$ (excluding $s$ and outgoing neighbors of $s$) maintains an $(n+1)$-entry vector $v_i$, where $v_i[t]~~(0 \leq t \leq n)$ is the estimate of output corresponding to the tested parameter $t$. In the beginning of the algorithm, every entry of vector $v_i$ is initialized to be a null value $\perp$, where $\perp$ is distinguished from all possible values of $x_s$.

Source node $s$ commits to its input $x_s$ at the start of the algorithm (round $0$), and transmits message $x_s$ to all its outgoing neighbors in round $1$. For the other nodes, the algorithm is as follows. 

\begin{itemize}
\item For outgoing neighbor of the source $s$:

\begin{enumerate}
\item In round $1$, it receives message $x$ directly from source $s$, and commits to output $x$.
\item In round $2$, it transmits messages $<x,0>, <x,1>,..., <x,n>$ to all its outgoing neighbors, and terminates.
\end{enumerate}

\item For node that is not an outgoing neighbor of $s$, in each round $r~(r>0)$:

\begin{enumerate}
\item For $0 \leq t \leq n$, each node $i$ that sets $v_i[t]$ in round $r-1$ to some value $x$, transmits message $<x,t>$ to all its outgoing neighbors.

\item For $0 \leq t \leq n$, through round $r$, if a node $i$ has received messages containing value $<x,t>$ from at least $t+1$ distinct incoming neighbors, then it sets $v_i[t] = x$.

\item In round $n$, each node $i$ commits to value $v_i[t']$, where $t'$ is the largest value in range $[0,n]$ such that $v_i[t'] \neq \perp$.
\end{enumerate}
Note that the algorithm performs $n$ rounds.
\end{itemize}

Now, we show that CPA-P is correct.

\begin{theorem}
Given a graph $G$ that can tolerate $f$-local faults, CPA-P achieves both validity and termination.
\end{theorem}

\begin{proof}
Denote by CPA-P-$t$ ($0\leq t \leq n$) the instance of CPA-P corresponding to the tested parameter $t$. Then by assumption of $G$, CPA-P-$f$ is correct. Thus, for each fault-free node $i$, $v_i[f]=x_s$, the input value at source $s$. Now, we prove the following claim:

\begin{claim}
\label{claim:valid-t}
For $t > f$, in CPA-P-$t$, fault-free nodes never decide on an invalid value, i.e., for each fault-free node $i$, either $v_i[t] = x_s$ or $v_i[t] = \perp$. 
\end{claim}

\begin{proof}
The proof is by induction.

\noindent {\em Induction basis:} Source node $s$ and its outgoing neighbors commit to output equal to the source's input $x_s$ in round 0 and 1, respectively. No other fault-free nodes commit in round 0 and 1.

\noindent{\em Induction:}
Suppose that $L$ is the set of fault-free nodes that have committed
to $x_s$ through round $r~(r>0)$. Thus, $s\in L$. Let $F'$ be the set of faulty nodes, and $|F'| = f$. Define $R=\sv-L-F'$. If $R=\emptyset$, then the proof is complete. Let us now assume that $R\neq\emptyset$.

Now consider round $r+1$.

Consider any fault-free node $u$ that has not committed prior to round $r+1$
(i.e., $u\in R$). All the nodes in $L$ have committed to $x_s$ by the
end of round $r$. Thus, in round $r+1$ or earlier, node $u$ may receive messages containing values different from $x_s$ only from nodes in $F'$. Therefore, node $u$ cannot commit to any value different from $x_s$ in round $r+1$, since by assumption $|N_u^- \cap F'| \leq f < t$.

Unlike the proof in Theorem \ref{thm:suf}, node $u$ may never gather enough (i.e., at least $t+1$) identical messages from its incoming neighbors, since $t > f$. Thus, for CPA-P-$t$, node $u$ may never terminate. In this case, $v_u[t] = \perp$.
\end{proof}

The source node $s$ and fault-free outgoing neighbors of $s$ commit to $x_s$ in round 0 and 1, respectively. By Claim \ref{claim:valid-t} and the fact that CPA-P-$f$ satisfies both validity and termination, each fault-free node $i$ (excluding $s$ and outgoing neighbors of $s$) commits to $x_s$. Thus, CPA-P is correct.

\end{proof}


\section{Discussion}

This section discusses some extensions on the result presented above.

\subsection{Generalized Fault Model}
\label{sec:general}

\newcommand{\Garrow}{\stackrel{g}{\Rightarrow}}

In this subsection, we briefly discuss how to extend the above results under
a generalized fault model. The generalized fault model \cite{tseng_iabc_generalized_fault_model} is characterized using {\em fault domain} $\scriptf \subseteq 2^{\scriptv}$ as follows:
Nodes in set $F$ may fail during an execution of the algorithm only if there exists set $F^*\in \scriptf$ such that $F\subseteq F^*$. Set $F$ is then said to be a {\em feasible} fault set.
\begin{definition}
Set $F\subseteq \scriptv$ is said to be a \,\underline{feasible}\, fault set, if there exists $F^*\in\scriptf$ such that $F\subseteq F^*$.
\end{definition}


For a set of nodes $B$, define $N^-(B) = \{ i~|~(i,j)\in \se,~i\not\in B,~j\in B\}$, the set of incoming neighbors of $B$.

\begin{definition}

Given $\scriptf$, for disjoint sets of nodes $A$ and $B$, where $B$ is non-empty.
\begin{itemize}
\item $A \Garrow B$ iff for every $F^* \in \scriptf$, $N^-(B) \cap A \not\subseteq F^*$.
\item $A\not\Garrow B$ iff $A\Garrow B$ is {\em not} true.
\end{itemize}
\end{definition}

~

Under the generalized fault model, step 3 of CPA needs to be modified as follows. Let us call the modified algorithm CPA-G.
\begin{list}{}{}{}
\item[3.] Through round $r$, if a node has received messages containing value $x$ from a set $M$, where $M$ is not a feasible fault set, then the node commits to value $x$. 
\end{list}

It is easy to show that a modified version of
Theorem \ref{thm:nec} stated below
holds for the generalized fault model. 

\begin{theorem}
\label{thm:nec:gen}
Suppose that CPA-G is correct in graph $G(\sv, \se)$ under the generalized
fault model. Let sets $F, L, R$ form a partition
of $\sv$, such that source (i) $s\in L$, (ii) $R$ is non-empty, and (iii) $F$ is a feasible fault set, then

\begin{itemize}
\item $L \Rightarrow R$, or

\item $R$ contains an outgoing neighbor of $s$, i.e., $N_s^+ \cap R \neq \emptyset$.

\end{itemize}

\end{theorem}

\subsection{Broadcast Channel}

We have so far assumed that the underlying network is a point-to-point network.
The results, however, can be easily extended to the {\em broadcast} or {\em radio model} \cite{Koo_broadcast, Vartika_broadcast} as well. In the {\em broadcast model}, when a node transmits a value, all of its outgoing neighbors receive this value identically. Thus, no node can transmit mismatching values to different outgoing neighbors. Then, it is easy to see that the same condition as the point-to-point network can be shown to be necessary
and sufficient for of CPA under the broadcast model as well.

Now consider the following variation of the CPA algorithm:
if the outgoing neighbors of source $s$ do not receive a message from $s$
in round 1, the message value is assumed to be some default value.
With this modification, the 
condition in Theorem \ref{thm:nec} can also be shown to be necessary and sufficient to perform Byzantine Broadcast \cite{PSL_BG_1980} under the broadcast model,
while satisfying the following three conditions (allowing $s$ to be faulty):
\begin{itemize}
\item {\bf Termination:} every fault-free node $i$ eventually decides on an output value $y_i$.

\item {\bf Agreement:} the output values of all the fault-free nodes are equal, i.e., there exists $y$ such that, for every fault-free node $i, y_i = y$.

\item {\bf Validity:} if the source node is fault-free, then for every fault-free node $i$, the output value equals the source's input, i.e., $y = x_s$.

\end{itemize}
The proof follows from the proof of Theorem \ref{thm:nec} and the observation
that if $s$ transmits a value, then all the outgoing neighbors of $s$ receive identical value from $s$, which equals its input $x_s$ when $s$ is fault-free.

\subsection{Asynchronous Network}

In our analysis so far, we have assumed that the system is synchronous.
For a point-to-point network with fault-free source $s$, 
it should be easy to see that the condition in Theorem \ref{thm:nec} is
also necessary and sufficient to achieve agreement using
a CPA-like under the asynchronous model \cite{AA_Dolev_1986} as well.
In this case, the algorithm may not proceed in rounds, but a node still
commits to value $x$ either on receiving the value directly from $s$,
or from $f+1$ nodes.

This claim may seem to contradict the FLP result \cite{FLP_one_crash}.
However, our claim assumes that the source node
is fault-free, unlike \cite{FLP_one_crash}.



\section{Conclusion}

In this paper, we explore broadcast in arbitrary network using the CPA algorithm in $f$-local fault model. In particular, we provide a {\em tight} necessary and sufficient condition on the underlying network for the correctness of CPA.

\section*{References}






\bibliographystyle{model1a-num-names}



\end{document}